\documentclass[10pt,conference]{IEEEtran}

\usepackage[dvips]{color}
\usepackage{tabu}
\usepackage{comment}
\usepackage{todonotes}
\usepackage{epsf}
\usepackage{epsfig}
\usepackage{times}
\usepackage{epsfig}
\usepackage{graphicx}
\usepackage{bbold}
\usepackage{mathtools}
\usepackage{mathrsfs}
\usepackage{amssymb,amsmath}
\usepackage{pdfpages}
\usepackage{epstopdf}
\usepackage{algorithm2e}
\newcommand{\Conv}{%
	\mathop{\scalebox{1.5}{\raisebox{-0.2ex}{$\circledast$}}
	}
}
\usepackage{dsfont}
\usepackage{lettrine} 
\usepackage{amsmath,epsfig,amssymb,algpseudocode,amsthm,cite,url}

\usepackage[justification=centering]{caption}
\usepackage{subcaption}
\usepackage{bbm}
\DeclarePairedDelimiter{\ceil}{\lceil}{\rceil}
\usepackage{csquotes}
\usepackage{amsmath,amssymb}

\DeclareFontFamily{OMX}{yhex}{}
\DeclareFontShape{OMX}{yhex}{m}{n}{<->yhcmex10}{}
\DeclareSymbolFont{yhlargesymbols}{OMX}{yhex}{m}{n}
\DeclareMathAccent{\wideparen}{\mathord}{yhlargesymbols}{"F3}
\usepackage{verbatim}
\usepackage{amsmath,amssymb}

\usepackage{verbatim}

\newtheorem{theorem}{\bf Theorem}
\let\oldtheorem\theorem
\renewcommand{\theorem}{\oldtheorem\normalfont}

\newtheorem{proposition}{\bf Proposition}
\let\oldproposition\proposition
\renewcommand{\proposition}{\oldproposition\normalfont}
\newtheorem{lemma}{\bf Lemma}
\let\oldlemma\lemma
\renewcommand{\lemma}{\oldlemma\normalfont}
\newtheorem{example}{Example}
\let\oldexample\example
\renewcommand{\example}{\oldexample\normalfont}

\newtheorem{definition}{\bf Definition}
\let\olddefinition\definition
\renewcommand{\definition}{\olddefinition\normalfont}
\newtheorem{remark}{Remark}
\let\oldremark\remark
\renewcommand{\remark}{\oldremark\normalfont}

\DeclareMathOperator*{\argmax}{argmax}

\usepackage{setspace}	
\def\tablename{Table}
\IEEEoverridecommandlockouts
\allowdisplaybreaks
\begin{document}
\title{\huge Ultra Reliable, Low Latency  Vehicle-to-Infrastructure Wireless Communications with Edge Computing}
\author{
\authorblockN{Md Mostofa Kamal Tareq$^1$, Omid Semiari$^1$, Mohsen Amini Salehi$^2$, and Walid Saad$^3$}\\\vspace*{-1em}
\authorblockA{\small $^{1}$Department of Electrical and Computer Engineering, Georgia Southern University, Statesboro, GA, USA,\\ 
		Email: \protect\url{{mt07092,osemiari}@georgiasouthern.edu}\\
$^{2}$School of Computing and Informatics, University of Louisiana at Lafayette, LA, USA,
Email: \protect\url{amini@louisiana.edu},\\
$^3$Wireless@VT, Bradley Department of Electrical and Computer Engineering, Virginia Tech, Blacksburg, VA, USA, Email: \protect\url{walids@vt.edu}
}
   \thanks{This research was supported by the U.S. National Science Foundation under Grants  IIS-1633363 and CNS-1739642.
   	}%
  }
%
\maketitle
\begin{abstract}
Ultra reliable, low latency vehicle-to-infrastructure (V2I) communications is a key requirement for seamless operation of autonomous vehicles (AVs) in future smart cities. To this end, cellular small base stations (SBSs) with edge computing capabilities can reduce the end-to-end (E2E) service delay by processing requested tasks from AVs locally, without forwarding the tasks to a remote cloud server. Nonetheless, due to the limited computational capabilities of the SBSs, coupled with the scarcity of the wireless bandwidth resources, minimizing the E2E latency for AVs and achieving a reliable V2I network is challenging. In this paper, a novel algorithm is proposed to jointly optimize AVs-to-SBSs association and bandwidth allocation to maximize the reliability of the V2I network. By using tools from \emph{labor matching markets}, the proposed framework can effectively perform distributed association of AVs to SBSs, while accounting for the latency needs of AVs as well as the limited computational and bandwidth resources of SBSs.  Moreover, the convergence of the proposed algorithm to a core allocation between AVs and SBSs is proved and its ability to capture interdependent computational and transmission latencies for AVs in a V2I network is characterized. Simulation results show that by optimizing the E2E latency, the proposed algorithm substantially outperforms conventional cell association schemes, in terms of service reliability and latency.  


 \vspace{-0cm}
\end{abstract}
\section{Introduction} \vspace{-0cm}
Autonomous vehicles (AVs) are among main transformative technologies in future smart cities. The deployment of AVs can help in reducing traffic congestions, increasing road safety, minimizing fuel consumption, and enhancing the overall driving experience \cite{7498068}. To effectively operate AVs, reliable vehicle-to-infrastructure (V2I) communications is required, particularly with widely deployed cellular base stations (BSs) to support connectivity and control for vehicles \cite{7992934}. In particular, cellular BSs can facilitate management of \emph{tasks} that AVs need to execute, by providing road information ahead of time, delivering high definition maps (HD-maps) for AVs, or maintaining coordination among AVs to prevent congestion \cite{8390386,8422983}. Nonetheless, most of the tasks associated with AVs are delay intolerant and require reliable processing with low latency. Therefore, the V2I wireless system must be capable of managing AVs' requested tasks under stringent latency and reliability  requirements \cite{7990497,magazine,mehdi}.

To minimize the V2I communications latency, \emph{edge computing} is an attractive solution that enables BSs to process AVs' requested tasks locally, without relying on  remote cloud servers \cite{7469991,7901477,AMINISALEHI201696}. However, several challenges must be addressed to seamlessly integrate edge computing with cellular V2I communications. \emph{First}, small cell BSs (SBSs) have limited computational resources and, hence, they can be easily \emph{overloaded} with AV tasks. \emph{Second}, the end-to-end (E2E) latency for AV task management in a V2I system depends on: a) uplink transmission latency (i.e., to request a task from the SBS); b) computational latency at the edge machine; and c) downlink transmission latency to send the processed task to the AV. Once coupled with heterogeneous task types and random wireless channel variations, optimizing the E2E quality-of-service (QoS) in V2I networks becomes very challenging and requires efficient \emph{AV-to-SBS association} jointly with \emph{wireless and computing resource management}.

Several works have been recently sought to address the aforementioned V2I challenges \cite{7990497,7828299,8004160,7994678,aaabbb}. In \cite{7990497}, the authors study the impact of transmission time interval (TTI) design on the performance of low-latency vehicular communications. In \cite{8004160}, the authors survey various software-defined latency control schemes in V2I networks. In \cite{7994678}, an edge computing framework is developed to reduce computational latency for vehicular services. The work in \cite{aaabbb} proposes different radio resource management methods for achieving low-latency vehicular communications. Meanwhile, most works on AV-SBS association (e.g., see \cite{6497017} and references therein) rely on conventional metrics such as maximum signal-to-interference-plus-noise ratio (max-SINR) and maximum received signal strength indicator (max-RSSI).


However, the prior art in \cite{7990497,7828299,8004160,7994678,aaabbb} studies computing and communication latency in isolation, rather than from an E2E perspective. For example, the work in \cite{7990497} considers a fixed value for the computational latency and neglects E2E latency. In addition, the authors in \cite{aaabbb}, focus solely on wireless resource management, without considering the computational latency. Moreover, existing works mostly rely on centralized resource management, while fast and efficient distributed algorithms are needed to manage tasks in dense V2I networks.


The main contribution of this paper is, thus, a novel low-latency V2I communications framework that maximizes the reliability of the V2I network by jointly optimizing AV-to-SBS association along with wireless resource management.  To this end, we build a novel  solution, based on \emph{matching theory}, that allows to account for the E2E latency requirements of AVs' tasks, as well as the limited computational and bandwidth resources of SBSs.   
To solve this problem, we propose a novel algorithm that iteratively associates AVs to SBSs, along with allocation of bandwidth. The proposed algorithm is proved to converge to a \emph{core allocation} between AVs and SBSs, thus guaranteeing the stability of the V2I network when using distributed implementations. Simulation results show that the proposed algorithm substantially improves the performance by maximizing the reliability of the V2I system and minimizing the E2E latency for AVs, compared with max-SINR and max-RSSI associations. 

The rest of this paper is organized as follows. Section II presents the system model. Section III presents the proposed algorithm. Simulation results are provided in section IV. Section V concludes the paper.

\section{System Model}
Consider a wireless cellular network composed of a set $\mathcal{N}$ of $N$ SBSs that are  distributed uniformly within a square area of size $A$. In this network, a set $\mathcal{M}$ of $M$ AVs are randomly deployed and must communicate with SBSs\footnote{We also consider road side units as SBSs.}.  Naturally,  seamless operation of AVs requires management of multiple tasks in real-time. Table \ref{tab1} summarizes a list of typical tasks with their latency requirements. To avoid latency at the backhaul network, SBSs are equipped with edge computing machines to process the AVs' requested tasks from the set $\mathcal{S}$, and send necessary information to the AVs. One example is HD-maps that  cannot be built by a single AV in real-time. In fact, an SBS can receive the location and sensing information from its associated AVs, process the data to build an HD-map, and send the HD-map to AVs. 
	\begin{table}[t!]
	\footnotesize
	\centering
	\caption{
		\vspace*{-0cm}Examples of V2I service latency requirements \cite{7990497}}\vspace*{-0.2cm}
	\begin{tabular}{|>{\centering\arraybackslash}m{2.4cm}|>{\centering\arraybackslash}m{2.5cm}|>{\centering\arraybackslash}m{2.2cm}|}
		\hline
		\bf{V2I Service} & \bf{Type}& \bf{Latency Requirement}  \\
		\hline
		Emergency Warning &  Safety &$100$ ms\\
		\hline
		See-through & Automated Driving& $50$ ms\\
		\hline
		Pre-crash Sensing Warning &Safety &$20$ ms\\
		\hline
		Automated Overtake &Automated Driving &$10$ ms \\
		\hline 
	\end{tabular}\label{tab1}\vspace{-2em}
\end{table}

Considering the edge computing capabilities of SBSs, the E2E latency for V2I communications between AV $m \in \mathcal{M}$ and an SBS $n \in \mathcal{N}$ can be defined as 
\begin{align}
	\tau(m,n;s) = \left[\tau_t(m,n;s)\right]T + \tau_p(m,n;s),
\end{align}
where $T$ is the duration of one TTI (in milliseconds) and $\tau_t(m,n;s)$ represents transmission latency, in terms of number of TTIs, for task $s$ of an AV $m$ associated with an SBS $n$. In addition, $\tau_p(m,n;s)$ is the computational latency at the edge unit of an SBS $n$ to process task $s$ of an AV $m$. Next, we characterize transmission and processing latencies in details.

\begin{figure}[t!]
	\centering
	\centerline{\includegraphics[width=8cm]{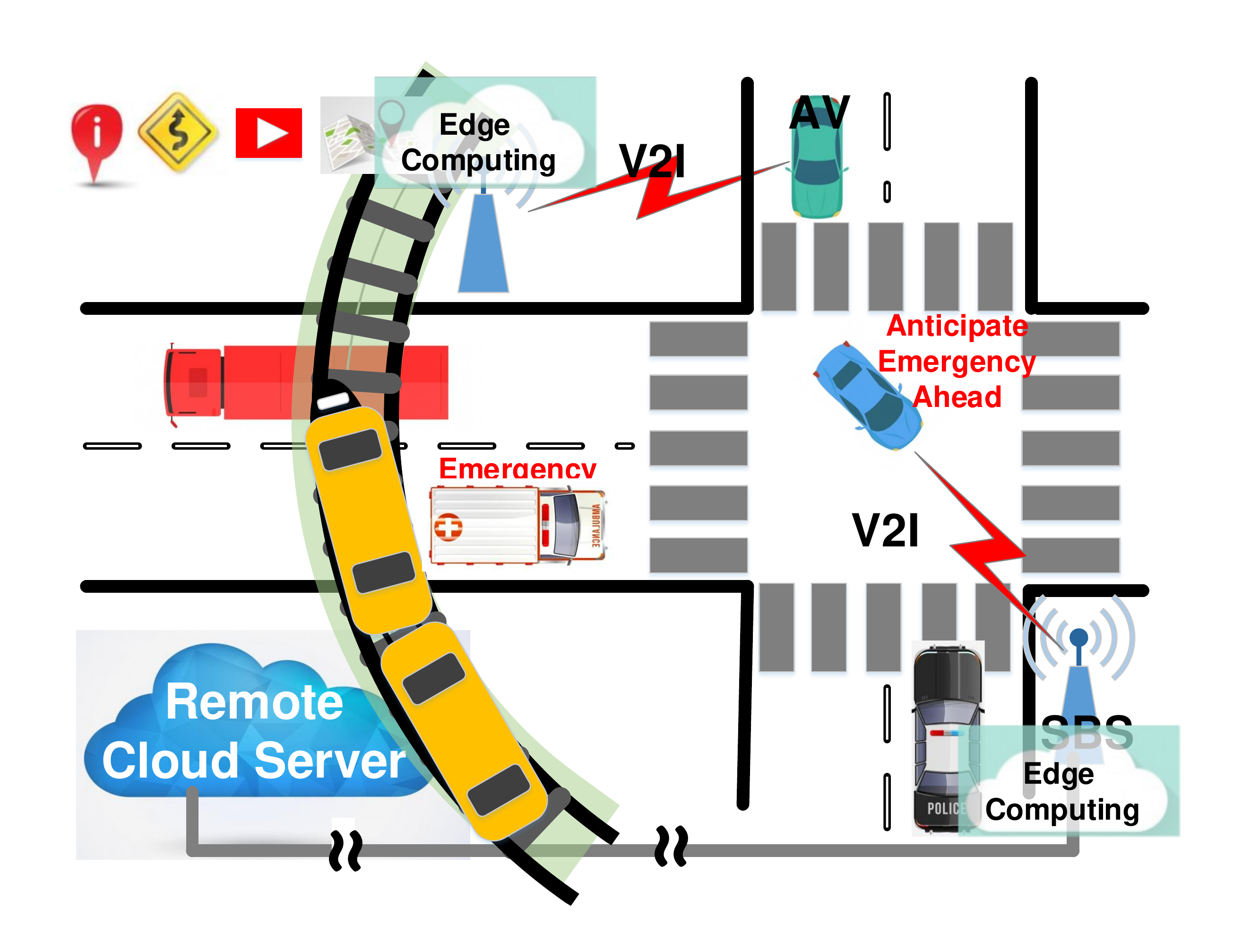}}\vspace{-1em}
	\caption{V2I network with edge computing capabilities.}\vspace{-2em}
	\label{model}
\end{figure}
\subsection{Wireless transmission latencies}
The overall transmission latency is the sum of both downlink and uplink transmission latencies and is given by:
\begin{align}\label{tau_t}
\tau_t(m,n;s) = \tau_d(m,n;s) + \tau_u(m,n;s),
\end{align} where $\tau_d(m,n;s)$ represents downlink transmission latency, in terms of number of TTIs, for AV $m$ to send its task $s$ to SBS $n$. This latency is given by:
\begin{align}\label{tau_d}
	\tau_d(m,n;s) =\ceil[\Big]{ \frac{I_d(s)}{R_d(m,n)T}}, 
\end{align}
where $\ceil{.}$ is a ceiling operation and $I_d(s)$ denotes the downlink packet size in bits, corresponding to task $s$. 
The denominator of \eqref{tau_d} represents the downlink data that can be transmitted within one TTI. For efficient utilization of time-frequency resources, we consider an orthogonal frequency-division multiplexing (OFDM) frame structure of bandwidth $W$, divided  equally into $K$ subchannels in a set $\mathcal{K}$, each with bandwidth $w$. Hence, the downlink data rate is:
\begin{align}\label{rate_d}
R_d(m,n) = w\sum_{k \in \mathcal{K}}y_{mnk}\log_2\left(1 + \gamma_d(m,n,k)\right),
\end{align}
where $y_{mnk}=1$, if subchannel $k$ is allocated by the SBS $n$ to AV $m$, otherwise $y_{mnk}=0$. $R_d(m,n)$ depends on the downlink SINR, $\gamma_d(m,n,k)$, given by:
\begin{align}\label{gamma_d}
\gamma_d(m,n,k) = \frac{G_mG_nP_nh_{mnk}L_{mn}}{\sum_{n' \neq n}{P_{mn'}} + \sigma_n^2},
\end{align}
where $P_n$, $P_{mn'}$, and $\sigma_n^2$ denote, respectively, the transmit power of SBS $n$, the received power from interfering SBS $n'$ at the AV $m$'s receiver, and the noise power. In \eqref{gamma_d}, $G_m$ and $G_n$ denote, respectively, the antenna gain for AV $m$ and SBS $n$. $h_{mnk}$, and $L_{mn}$ represent, respectively, the Rayleigh fading channel gain at subchannel $k$, and path loss of the downlink between AV $m$ and SBS $n$.  In low-latency V2I, the TTI duration is considered substantially small (in the order of one OFDM symbol duration \cite{7990497}), and, thus, the channel gain $h_{mnk}$ (and the rate in \eqref{rate_d}) can be  considered constant within the course of transmitting one packet.

Analogous to \eqref{tau_d}, we can find the uplink transmission latency $\tau_u(m,n;s)$, using the uplink SINR:
\begin{align}\label{gamma_u}
\gamma_u(m,n,k') = \frac{G_mG_nP_mh_{mnk'}L_{mn}}{\sum_{m' \neq m}{P_{m'n}} + \sigma_n^2},
\end{align}
where $P_m$ is the transmit power of the AV $m$ and $P_{m'n}$ is the received power from an interfering AV $m'$ that transmits over subchannel $k'$.  Given that the uplink traffic in V2I mainly has small-size packets (e.g., to request a service or send a short control packet), we assume that uplink transmissions are managed within one subchannel $k' \notin \mathcal{K}$. 


\subsection{Computational latency in an SBS}
The computational latency, also known as the execution delay, for managing a task at an edge computing machine is a random variable. Such randomness mainly stems from the fact that the execution time of a task can depend on the data to be processed \cite{AMINISALEHI201696}.  For example, the execution time for compressing an HD-map depends on the quality or level of details in the map. Furthermore, we consider a heterogeneous type of machines across different SBSs, i.e., a single task may require different computational latency, once processed at different SBSs. Due to this uncertainty, it is common to define a probability mass function (pmf) for the computational latency of an arbitrary task $s$ being completed by the edge machine of SBS $n$ by $t_k$ number of time steps\footnote{Although processing latency is a continuous random variable, this metric is commonly quantized into small time steps in a machine~\cite{AMINISALEHI201696}.}~\cite{AMINISALEHI201696}.

Without loss of generality, we consider one edge computing machine per SBS. In a V2I scenario, multiple tasks can be assigned to an SBS. In this case, the tasks assigned to a machine (of SBS) $n$  are batched at the queue $\mu(n) \subseteq\mathcal{S}$. Here, we note that the completion time pmf of an arbitrary task $s$ depends on other tasks in $\mu(n)$. We can now find the completion time pmf of a queued task $s \in \mu(n)$, $\mathbbm{P}_{s,n}(t_k;\mu(n))$, by convolving the execution time  pmf of tasks ahead of $s$ in the queue $\mu(n)$ \cite{AMINISALEHI201696}:
\begin{align}\label{pmf}
&\mathbbm{P}_{s,n}(t_k;\mu(n)) = \Conv_{s'\in \mu_s(n,t_k)}\mathbbm{P}_{s',n}(t_k),\notag\\ 
&=\mathbbm{P}_{s_1,n}(t_k) \circledast \mathbbm{P}_{s_2,n}(t_k) \circledast \cdots \circledast \mathbbm{P}_{s_{|\mu_s(n,t_k)|},n}(t_k),
\end{align}
where $\circledast$ denotes the convolution operation and $\mu_s(n,t_k)=\{s_1,s_2, \cdots s_{|\mu_s(n,t_k)|}\}$ is the set of all tasks ahead of $s$ in queue of machine $n$. Using the completion time  pmf in \eqref{pmf}, the SBS can find the expected completion time for its associated tasks:
\begin{align}\label{tau_p_expected}
\tau_c(\mu(n),n) &= \sum_{m \in \mu(n)}\mathbbm{E}\left[\tau_p(m,n;s)\right], \notag\\
&= \sum_{m \in \mu(n)}\sum_{t_k=1}^{t_{\text{max}}^s}t_k\mathbbm{P}_{s,n}(t_k;\mu(n)),
\end{align} 
where $t_{\text{max}}^s$ is the maximum number of processing time steps before task $s$ is dropped. 
\subsection{Problem Formulation}
Considering the proposed framework, the goal is to maximize the reliability of the V2I system, while considering the latency requirements of each AV's task. With this in mind, we define the Bernoulli random variable $\kappa$, such that
\begin{align}\label{y}
\kappa(m,n;s) = \begin{cases}
1, &\text{if} \,\,\,\,\tau(m,n;s) \leq \tau_{\text{th}}(s),\\
0, &\text{otherwise},
\end{cases}
\end{align}
where $\tau_{\text{th}}(s)$ is the tolerable E2E delay for task $s$. That is, $\kappa(m,n;s)=1$, if task $s$ is successfully managed within a time less than its E2E latency constraint, otherwise, $\kappa(m,n;s)=0$. Using \eqref{y} and the 3GPP definition for reliability \cite{3gpp3}, we can define the V2I system reliability as
\begin{align}\label{reliability}
\!\!\!\eta\left(\boldsymbol{x},\left[\boldsymbol{w}_1, \boldsymbol{w}_2, \cdots, \boldsymbol{w}_N\right]\right) \!=\! \frac{1}{M}\sum_{n \in \mathcal{N}}\sum_{m \in \mathcal{M}} \!\!x_{mn}\kappa(m,n;s),
\end{align}
with $\boldsymbol{w}_n = \left[w_{1n}, w_{2n}, \cdots w_{Mn}\right]$ where each element is the bandwidth allocation variable $w_{mn}=w\sum_{k}y_{mnk}$ for SBS $n$. Meanwhile, $\boldsymbol{x}$ denotes the AV-SBS association vector with elements $x_{mn}=1$, if AV $m$ is associated with the SBS $n$, otherwise $x_{mn}=0$. 
In fact, \eqref{reliability} implies that the network reliability depends on: 1) Association of AVs to SBSs; and 2) Each SBS's bandwidth allocation. Depending on the subset of AVs  associated with an SBS, the computational latency will change, as seen in \eqref{pmf}. Thus, to maximize the reliability $\eta$, the challenge is to jointly find an optimal AV-SBS association and resource allocation, while considering the fact that the transmission and computational latencies for one AV is affected by other AVs. 

Therefore, we formulate the joint AV-SBS association and resource allocation problem as: 
\begin{align}
\argmax_{\boldsymbol{x},\left[\boldsymbol{w}_1, \boldsymbol{w}_2, \cdots, \boldsymbol{w}_N\right]} & \,\,\,\,\,\,\eta\left(\boldsymbol{x},\left[\boldsymbol{w}_1, \boldsymbol{w}_2, \cdots, \boldsymbol{w}_N\right]\right),\label{opt1}\\
\text{s.t.,}&\sum_{n\in\mathcal{N}}x_{mn}\leq 1, \forall m \in \mathcal{M},\label{opt2}\\
&\sum_{m\in\mathcal{M}}x_{mn}\leq M, \forall n \in \mathcal{N},\label{opt3}\\
&\sum_{m\in\mathcal{M}}\sum_{k\in\mathcal{K}}x_{mn}y_{mnk}\leq K, \forall n \in \mathcal{N},\label{opt4}\\
&\sum_{m\in\mathcal{M}}x_{mn}y_{mnk}\leq 1, \forall n \in \mathcal{N}, \label{opt5}\\
&x_{mn},y_{mnk} \in \{0,1\}.\label{opt6}
\end{align}
Constraints \eqref{opt2} and \eqref{opt3} imply that each AV is associated to at most one SBS, and each SBS can serve up to $M$ AVs. Moreover, \eqref{opt4} indicates that $K$ subchannels can be used  by each SBS for downlink transmissions, while \eqref{opt5} ensures orthogonal subchannel allocation within each cell. The reliability maximization problem with joint AV-SBS association and resource allocation in \eqref{opt1}-\eqref{opt6} is an optimization problem with minimum unsatisfied relations \cite{Amaldi98} that is NP-hard and difficult to solve. Next, we propose a novel framework to solve this problem.  

\section{Matching Theory for Low-Latency V2I Communications}

To jointly optimize AV-SBS association and bandwidth allocation in \eqref{opt1}-\eqref{opt6}, 
one must find a fast-converging algorithm that can manage the network resources efficiently for dense V2I networks having a large number of AVs and SBSs. To this end, we build our solution based on \emph{matching theory}, a mathematical framework that can yield efficient algorithms for solving combinatorial assignment problems, such as the problem in \eqref{opt1}-\eqref{opt6}. In particular, we model the problem of AVs to SBSs association by using the analogous problem of workers to firms assignment in \emph{labor matching markets} \cite{matchinglabor}. That is, considering a set of workers (AVs) and a set of firms (SBSs), the goal of each worker is to be hired with maximum possible salary, while firms aim to hire a subset of workers that maximize their revenue. Using this two-sided framework, next, we show how the V2I resource management problem can be formulated as a labor matching market.

\subsection{Utility Functions for SBSs and AVs}
To enable low-latency communications in V2I networks, each SBS $n$ aims to select a subset of AVs $\mathcal{M}^n \subseteq \mathcal{M}$ as well as a bandwidth allocation that maximize the following utility (objective) function:
\begin{align}\label{utility1}
	\!\!\!U_n(\mathcal{M}^n; \boldsymbol{w}_n) \!\!=\!\!\!\!\sum_{m \in \mathcal{M}^n} \!\!\left[\frac{\alpha }{w_{mn}}-\tau_t(m,n;s)\right]\! \!-\! \tau_c(\mathcal{M}^n,n), 
\end{align}
where the transmission latency $\tau_t$ is given in \eqref{tau_t}. Here, we note that the computational latency for tasks of AVs in $\mathcal{M}^n$ is random and  unknown a priori. Therefore, in \eqref{utility1}, the SBS considers the expected value of the completion time for the subset of AVs, as per \eqref{tau_p_expected}. Moreover, the first term in \eqref{utility1} is inversely proportional to the bandwidth allocated to an AV $m$, and $\alpha$ is a control parameter. The allocated bandwidth can be seen as the cost of serving an AV and, thus, the first term in \eqref{utility1} will prevent unfair bandwidth allocations. In fact, analogous to the labor matching market, the first term in  \eqref{utility1} is the salary that a firm has to pay for hiring a worker. 

Meanwhile, each AV aims to be associated with an SBS that can minimize its E2E latency. Nonetheless, an AV does not know the computational latency for its task at an SBS. That is because the computational latency of one AV depends on other AVs associated with the same SBS, as captured in \eqref{pmf}. This information can be provided by the SBS. Hence, the utility that an AV $m$ assigns to an SBS $n$ is:
\begin{align}\label{utility2}
 U_m(n;w_{mn})=-\tau_t(m,n;s) T -\mathbbm{E}\left[\tau_p(m,n;s)\right],
\end{align}
where $\tau_t$ depends on the resource allocation vector $\boldsymbol{w}_n$. 

\subsection{V2I Communications as a Matching Problem}
Using the defined utility functions for AVs and SBSs, We define the proposed matching problem as follows:
\definition{An AV-SBS matching is a relation $f: \mathcal{M}\rightarrow \mathcal{N}$} that  satisfies:
\begin{itemize}
	\item[1)] For any AV $m$, $f(m) \in \mathcal
	N \cup \{m\}$. In fact, $f(m)=m$ implies that AV $m$ is not assigned to any SBS. 
	\item[2)] For any SBS $n$, $f(n)=\mathcal{M}^n \subseteq \mathcal{M}$.
	\item[3)] $f(m)=n$, if and only if $m \in f(n)$.
\end{itemize}
This definition also ensures meeting the feasibility constraints in \eqref{opt2} and  \eqref{opt3}.

Furthermore, subject to a bandwidth allocation vectors $\left[\boldsymbol{w}_1, \boldsymbol{w}_2, \cdots, \boldsymbol{w}_N\right]$, a matching $f$, denoted by a pair $(f;\left[\boldsymbol{w}_1, \boldsymbol{w}_2, \cdots, \boldsymbol{w}_N\right])$, is called \emph{individually rational}, if it meets the following conditions: 1) For any AV assigned to an SBS $f(m)$, $w_{mf(m)}>0$; and 2) For any SBS $n$, $ 0 < U_n(f(n);\boldsymbol{w}_n) < \infty$.
The first condition implies that at least one subchannel must be allocated to the AV $m$ by its assigned SBS $f(m)$, otherwise, the AV will be indifferent between being assigned to the SBS $f(m)$ or not. The second condition implies that the  utility of an SBS must be nonnegative and finite, otherwise, there is no need to allocate resources to the AVs in $f(n)$. Our goal is to find a strict core allocation of AVs to SBSs, as defined next.
\definition{An individually rational AV-SBS matching $(f;\boldsymbol{w})$ is a \emph{core allocation}, if there are no pair of SBS-subset of AVs $(n,\mathcal{M}')$ with a bandwidth allocation vector $\boldsymbol{\hat{w}}_n$, that satisfy,
\begin{itemize}
	\item[1)] $U_m(n;\hat{w}_{mn})\!>\! U_m(f(m),w_{mf(m)})$, for all $m \in \mathcal{M}'$, and
	\item[2)] $U_n(\mathcal{M}';\boldsymbol{\hat{w}}_n) > U_n(f(n);\boldsymbol{w}_n)$.
\end{itemize}
In fact, a pair $(n,\mathcal{M}')$ that meets above two conditions can improve their utility by \emph{blocking} the matching $f$ and making a new allocation. In particular, the notion of core allocation guarantees \emph{stability} of the V2I system by preventing undesired SBS-AVs allocations.}

 Nonetheless, finding a core AVs-to-SBSs allocation is challenging, due to the interdependent utilities of the AVs and SBSs that stem from two facts: 1) From \eqref{pmf}, the processing latencies of the assigned AVs to the same SBS are interrelated; and 2) From \eqref{opt4}, allocated bandwidth to one AV depends on the resource allocation to other AVs within the same cell. In fact, classical methods such as the deferred acceptance algorithm \cite{79} fail to yield a core allocation for the V2I problem \cite{matchinglabor}. Thus, we next propose a new algorithm that guarantees finding a core allocation of AVs to SBSs. 
 
 \section{Proposed Matching Algorithm For Joint AV-SBS Allocation and Resource Management}
 
 The key idea for guaranteeing the core allocation is to allow \emph{negotiations} for bandwidth between AVs and SBSs, while performing the AV-SBS association. That is, SBSs can offer a certain E2E latency to each AV (by allocating a number of subchannels), while the AV can accept the offer or reject it for a better allocation with another SBS.

 Building on this idea, we propose a novel algorithm in Table \ref{algo1} that proceeds as follow: Initially at round $j=0$, each SBS allocates one subchannel to each AV. As the algorithm proceeds, each SBS $n$ updates the bandwidth allocation vector, $\boldsymbol{w}_n(j)=[w_{1n}(j),w_{2n}(j),\cdots, w_{Mn}(j)]$, at round $j$, using the following rule: If an AV rejects the offer by an SBS $n$ in round $j-1$, then $w_{mn}(j)=w_{mn}(j-1)+w$; otherwise, $w_{mn}(j)=w_{mn}(j-1)$.

 At any round $j$, subject to the bandwidth allocation $\boldsymbol{w}_n(j)$, each SBS $n \in \mathcal{N}$ selects a subset of AVs that maximizes its utility in \eqref{utility1}. The process of AV selection is performed in Step 2. Subsequently in Step 3, each SBS offers association to its selected AVs in Step 2. Any offer in round $j-1$ that was not rejected will be repeated in round $j$. In Step 4, each AV tentatively accepts the offer that maximizes its utility in \eqref{utility2} and rejects the rest. Finally in Step 5, SBSs update their bandwidth allocation vectors $\boldsymbol{w}_n(j), \forall n \in \mathcal{N}$, according to the rule explained previously. The algorithm converges once no offer is rejected by the AVs. Prior to proving the convergence of the proposed algorithm to a core allocation, we make the following preliminary observations:
 \begin{remark}\label{lemma1}
 	Every AV has at least one association offer in each round. 
 \end{remark}This can be easily verified by noting that at Step 1, each SBS extends an association offer to all AVs. Since at any round, each AV tentatively accepts one offer, the AV's allocated bandwidth remains constant. Moreover, from Step 3, we note that any offer that is not rejected must be repeated in the next round. Therefore, at any round, AVs have at least one association offer.

	\begin{table}[t!]
	\footnotesize
	\centering
	\caption{
		\vspace*{-0cm}Proposed AV-SBS Association and Resource Management Algorithm }\vspace*{-0.2cm}
	\begin{tabular}{p{8 cm}}
		\hline \vspace*{-0em}
		\textbf{Inputs:}\,\,$\mathcal{M}$, $\mathcal{N}$, $\mathcal{K}$, $\mathcal{S}$.\\
		\hspace*{1em}
		\textbf{Step 1:} Let $t=0$. Each SBS allocates one subchannel to each AV. Each SBS sends proposal to all AVs, notifying them of their E2E latency, according to \eqref{utility2}.\\
		\hspace*{-0.5em} \While{there are proposal rejections}{

			\textbf{Step 2:} At each round $j$, each SBS $n$ selects a subset of AVs $\mathcal{M}^n$ that maximize the utility $U_n(\mathcal{M}^n;\boldsymbol{w}_n(j))$, with the bandwidth allocation vector $\boldsymbol{w}_n(j)=[w_{1n}(j),w_{2n}(j),\cdots, w_{Mn}(j)]$, where $w_{mn}=w\sum_{k}y_{mnk}$. Each SBS $n$ sorts AVs in descending order according to a utility $U_n(m)=\frac{\alpha}{w_{mn}}-\tau_t(m,n,s)$. The SBS adds the first AV from the list to $\mathcal{M}^n$, calculates $U_n(\mathcal{M}^n,\boldsymbol{w}_n(t))$, and while this utility is positive, adds other AVs from the ordered list one by one.

			\textbf{Step 3:} Each SBS offers association to all AVs selected in Step 3. Any offer in round $t-1$ that was not rejected will be repeated in round $t$.
			
			\textbf{Step 4:} Each AV that receives one or more offers rejects all, except the one that maximizes its utility in \eqref{utility2}.
			
			\textbf{Step 5:} If an AV rejects the offer by an SBS $n$ in round $t-1$, let $w_{mn}(t)=w_{mn}(t-1)+w$; otherwise, $w_{mn}(t)=w_{mn}(t-1)$.
			
			$t+1\gets t$
		} 			
		\hspace*{0em}\textbf{Output:}\,\,Strict core allocation $f^*$\vspace*{0em}\\
		\hline
	\end{tabular}\label{algo1}\vspace{-2em}
\end{table}

 \begin{lemma}\label{lemma2}
	Each AV will have exactly one offer after a finite number of rounds and the algorithm converges.
\end{lemma}
 \begin{proof}
 	From Lemma \ref{lemma1}, we note that each AV has at least one offer at each round. Moreover, according to the bandwidth allocation rule in Step 5, the bandwidth must be increased for the AV by all proposing SBSs, except the one that its offer is accepted. Meanwhile, from \eqref{utility1}, the utility of SBSs is a decreasing function of the allocated bandwidth. Thus, as algorithm proceeds, the number of offers for each AV decreases until each AV receives only one offer that it accepts. Since no rejection is made at that point, the algorithm converges.  
 \end{proof}
 
  \begin{lemma}\label{lemma3}
 	The proposed algorithm in Table \ref{algo1} converges to an individually rational allocation of AVs and SBSs. 
 \end{lemma}
  \begin{proof}
 Let $(f^*,\left[\boldsymbol{w}^*_1, \boldsymbol{w}^*_2, \cdots, \boldsymbol{w}^*_N\right])$ be the outcome at the convergence point of the algorithm after $j^*$ iterations. Since in Step 1, at least one subchannel is allocated to AVs, and from Step 5, the bandwidth allocation increases at each round, then, $w_{mf^*(m)}(j^*)>0$ for all $m \in \mathcal{M}$. From the AV selection process in Step 2, the AV subset $\mathcal{M}^n$ selected by an arbitrary SBS $n$ will yield a finite, positive utility $U_n(\mathcal{M}^n;\boldsymbol{w}_n)$. Thus, after $j^*$ iterations, $U_n(f^*(n),\boldsymbol{w}_n^*)$ will be positive, which concludes the proof.
 \end{proof}
 
\theorem{The proposed algorithm in Table \ref{algo1} is guaranteed to converge to a core association of AVs and SBSs.}
\begin{proof}
Lemma \ref{lemma2} shows the convergence of proposed algorithm after an arbitrary $j^*$ number of iterations. Let the outcome of algorithm be $(f,\left[\boldsymbol{w}_1(j^*),\boldsymbol{w}_2(j^*),\cdots, \boldsymbol{w}_N(j^*)\right])$. We prove the core allocation of the outcome by contradiction. That is, suppose that this outcome is not a strict core allocation. Since from Lemma \ref{lemma3}, $(f,\left[\boldsymbol{w}_1(j^*),\boldsymbol{w}_2(j^*),\cdots, \boldsymbol{w}_N(j^*)\right])$ is individually rational, there must be a blocking pair of an SBS and a subset of AVs, $(n,\mathcal{M}')$, with a bandwidth allocation vector $\boldsymbol{\hat{w}}_n$, such that, $\forall m \in \mathcal{M}'$:
\begin{align}
&U_m(n;\hat{w}_{mn}) > U_m(f(m),w_{mf(m)}(j^*)),\, \text{and} \label{cond1}\\
&U_n(\mathcal{M}';\boldsymbol{\hat{w}}_n) > U_n(f(n);\boldsymbol{w}_n(j^*)).\label{cond2}
\end{align}
From \eqref{cond1}, any AV $m \in \mathcal{M}'$ must never have received an offer from the SBS $n$ with bandwidth allocation $\hat{w}_{mn}$ (or greater than $\hat{w}_{mn}$) at any round of the algorithm. Otherwise, the AV $m$ would have accepted that offer. Meanwhile, since the allocated bandwidth to each AV increases or remains constant after each round (according to the update rule in Step 5), then for an SBS $n$ to form a blocking pair with AVs in $\mathcal{M}'$, the bandwidth allocation must satisfy $\hat{w}_{mn} \geq w_{mf(m)}$, for all $m \in \mathcal{M}'$. However, given that the utility of SBSs is a decreasing function of the allocated bandwidth, then,
\begin{align}\label{eq_proof_3}
	U_n(\mathcal{M}'; \boldsymbol{w}_n(j^*)) \geq U_n(\mathcal{M}'; \boldsymbol{\hat{w}}_n)\!>\!U_n(f(n);\boldsymbol{w}_n(j^*)),
\end{align} 
where the strict inequality in \eqref{eq_proof_3} directly results from \eqref{cond2}. In fact,  \eqref{eq_proof_3} implies that the SBS $n$ will propose to the subset of AVs in $\mathcal{M}'$ with bandwidth allocation vector $ \boldsymbol{w}_n(j^*)$ which  contradicts the initial assumption for the convergence of the algorithm. Therefore, such a blocking pair does not exist and thus,  convergence to a core allocation is guaranteed. 
\end{proof}
\begin{table}[t!] 
	\footnotesize
	\centering
	\caption{\vspace*{-0cm}  Simulation Parameters}\vspace*{-0.2cm}
	\begin{tabular}{|>{\centering\arraybackslash}m{2.4cm}|>{\centering\arraybackslash}m{2.5cm}|>{\centering\arraybackslash}m{2.2cm}|}
		\hline
		\bf{Notation} & \bf{Parameter} & \bf{Value} \\
		\hline
		$N$ & Number of SBSs & $10$\\
		\hline
		$M$ & Number of AVs & $10$ to $40$\\
		\hline
		$P_n$ & Transmit power of an SBS & $100$ mW\\
		\hline
		$P_m$ & Transmit power of an AV & $10$ mW\\
		\hline
		$W$ & System bandwidth & $100$ MHz\\
		\hline
		$I_d$ & Downlink packet size & $5$ kbits\\
		\hline
		$I_u$ & Uplink packet size & $100$ bits\\
		\hline
		$G_n,G_m$ & Antenna gains & $1$\\
		\hline
		$\sigma_n^2$ & Noise power& $-90$ dBm\\
		\hline
		$\alpha$ & Control parameter& $20$k\\
		\hline
		$w$ & Bandwidth of subchannel& $180$kHz\\
		\hline
		
	\end{tabular}\label{tabsim1}\vspace{-2em}
\end{table}
\section{Performance Evaluation}

\subsection{Simulation Parameters}
We consider an area of size $100$~m $\times$ $100$~m with AVs and SBSs located randomly across the area. We consider $10$ SBSs,  while the number of AVs varies from $10$ to $40$. Statistical results are averaged over large number of independent runs. Simulation parameters are summarized in Table \ref{tabsim1}. 

Furthermore, we assign a random task to each AV from  three task types in the set $\mathcal{S}=\{s_1,s_2,s_3\}$. Depending on the edge computing capabilities at an SBS,  we consider two types of SBSs with different latency distributions to manage the tasks in $\mathcal{S}$. For each task $s_i \in \mathcal{S}$ processed at a machine type $j\in \{1,2\}$, the pmf of the computational latency follows Gaussian distribution, $N(\mu_{ij}, \sigma_{ij}^2)$, with mean $\mu_{ij}$ and variance $\sigma^2_{ij}$ \cite{AMINISALEHI201696}, specified in Table \ref{tabsim2}. The tolerable E2E latency for each task type is also specified in Table \ref{tabsim2}. We compare the performance of our proposed algorithm with both Max-SINR and Max-RSSI associations.


\begin{table}[t!] 
	\footnotesize
	\centering
	\caption{\vspace*{-0cm}Mean and Standard Deviation of the Computational Latency Distribution}\vspace*{-0.1cm}
	\begin{tabular}{|>{\centering\arraybackslash}m{.9cm}|>{\centering\arraybackslash}m{1.6cm}|>{\centering\arraybackslash}m{1.6cm}|>{\centering\arraybackslash}m{1.6cm}|>{\centering\arraybackslash}m{1.5cm}|}
		\hline
		 &  {Task Type 1 $\tau_{\text{th}} = 20$ ms} & {Task Type 2 $\tau_{\text{th}} = 50$ ms}  & {Task Type 3 $\tau_{\text{th}} =\! 100$ ms} \\
		\hline
		Machine Type 1 & $(\mu_{11},\sigma_{11}) = (1,0.5)$ & $(\mu_{21},\sigma_{21}) = (2,0.5)$ & $(\mu_{31},\sigma_{31}) = (5,0.5)$\\
		\hline
	Machine Type 2 & $(\mu_{12},\sigma_{12}) = (2,0.5)$&$(\mu_{22},\sigma_{22}) = (4,0.5)$ &$(\mu_{32},\sigma_{32}) = (10,0.5)$\\
		\hline
		
	\end{tabular}\label{tabsim2}\vspace{-1em}
\end{table}

\subsection{Simulation Results}

\begin{figure}[t!]
	\centering
	\centerline{\includegraphics[width=6cm]{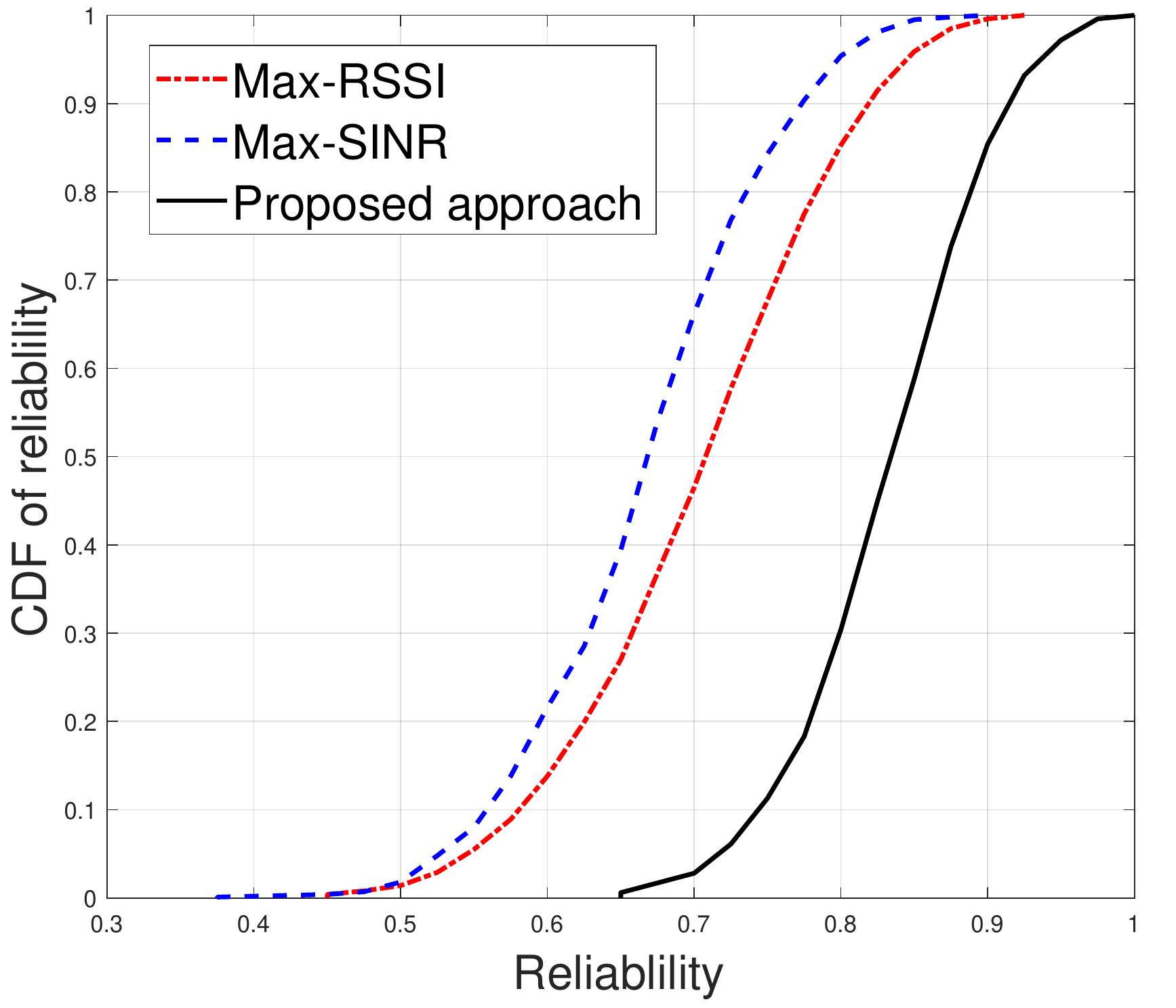}}\vspace{-0.2cm}
	\caption{\small CDF of the reliability.}\vspace{-.4cm}
	\label{fig1}
\end{figure}

Fig. \ref{fig1} shows the cumulative distribution function (CDF) of the reliability in the V2I network with $M=40$ AVs and $N=10$ SBSs. The results in Fig.  \ref{fig1} show that the proposed algorithm significantly outperforms the max-SINR and max-RSSI schemes. For example, the probability of achieving reliability less than $0.8$ is only $30\%$ in the proposed scheme, while this probability for the max-SINR and max-RSSI is $95\%$ and $85\%$, respectively. Such performance gain is mainly due to accounting for the E2E latency, while performing AV-to-SBS association and bandwidth allocation, while the baseline schemes do not optimize the E2E latency.

\begin{figure}[t!]
	\centering
	\centerline{\includegraphics[width=6.2cm]{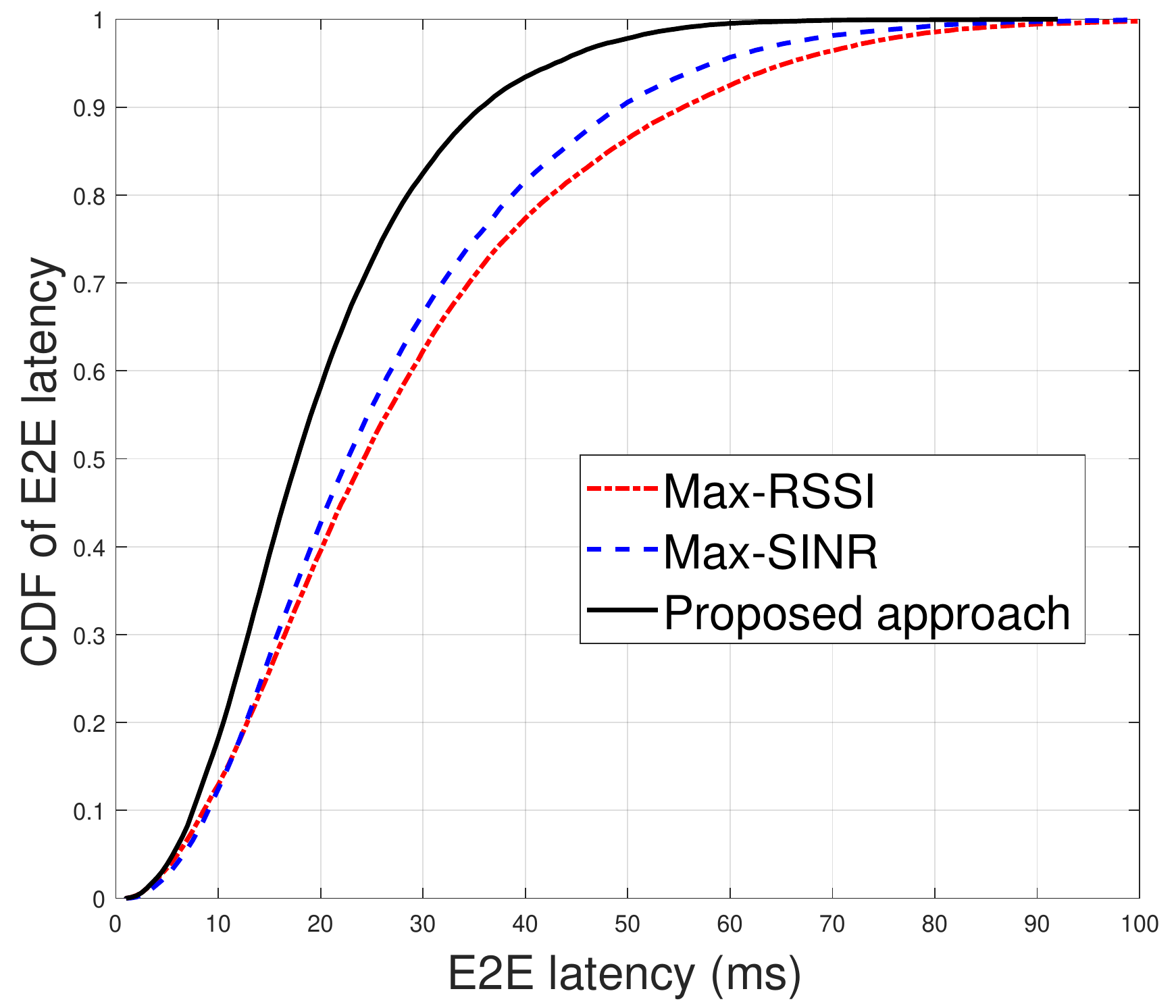}}\vspace{-0.2cm}
	\caption{CDF of the E2E latency.}\vspace{-.4cm}
	\label{fig2}
\end{figure}

In Fig. \ref{fig2}, we show the CDF of the E2E latency and compare the performance for the three approaches  in a V2I network with $M=40$ AVs and $N=10$ SBSs. First, we can observe that the proposed scheme can guarantee $50$ ms E2E latency with a high probability close to $99\%$. However, both baseline approaches can only satisfy this E2E latency requirement with probabilities less than $90\%$. For a large V2I network with $M=40$ AVs, the results in Fig. \ref{fig2} show that the proposed algorithm can effectively minimize the E2E latency. Clearly, the E2E latency will reduce as the network load decreases.
\begin{figure}[t!]
	\centering
	\centerline{\includegraphics[width=6.2cm]{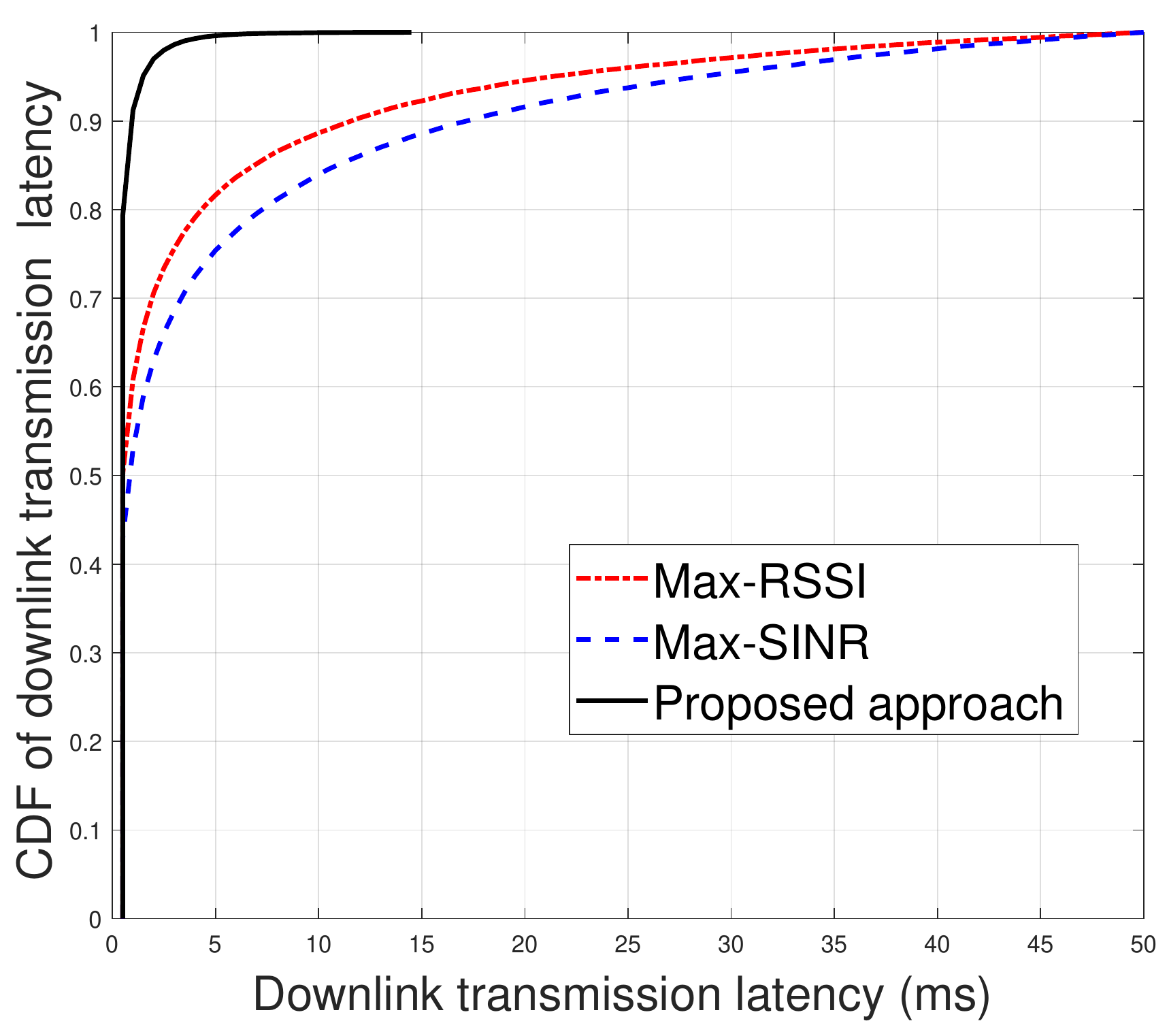}}\vspace{-0.2cm}
	\caption{CDF of the downlink transmission latency.}\vspace{-.4cm}
	\label{fig3}
\end{figure}

To show the impact of both transmission and computational  latencies on the overall E2E latency, the CDF of these metrics are shown, respectively, in Fig. \ref{fig3} and \ref{fig4}, for $M=40$ AVs and $N=10$ SBSs. Comparing the values for the transmission latency in Fig. \ref{fig3}, with computational latency in  Fig. \ref{fig4}, we can observe that the computational latency at the edge computing machine is substantial (can be up to $100$ ms) and cannot be neglected. Moreover, Fig. \ref{fig4} shows that the proposed scheme yields more efficient AV-to-SBS association, compared with the baseline algorithms, as it accounts for the computational latency and the amount of load at each edge machine. This feature can be viewed as load balancing, where instead of taking the number of AVs into account, our approach considers the computational loads of the assigned tasks at each SBS.

In Fig. \ref{fig5}, we show the average downlink data rate per AV, versus the network size. Clearly, the average rate per AV decreases as more AVs exist in the network. The results in Fig. \ref{fig5} show that the proposed algorithm outperforms the baseline approaches substantially in terms of data rate. For instance, the performance gains for a V2I network with $M=20$ AVs are $49\%$ and $90\%$, respectively, compared with the max-RSSI and max-SINR schemes. 

Finally, Fig. \ref{fig6} shows the number of iterations (with $95\%$ confidence error bars) of the proposed algorithm, versus the number of AVs.  The results in Fig. \ref{fig6} demonstrate that,  even for large V2I networks with $10$ SBSs and $30$ AVs, the number of iterations will not exceed $60$. Moreover, the results show that the number of iterations is polynomial with respect to the network size.

\begin{figure}[t!]
	\centering
	\centerline{\includegraphics[width=6.2cm]{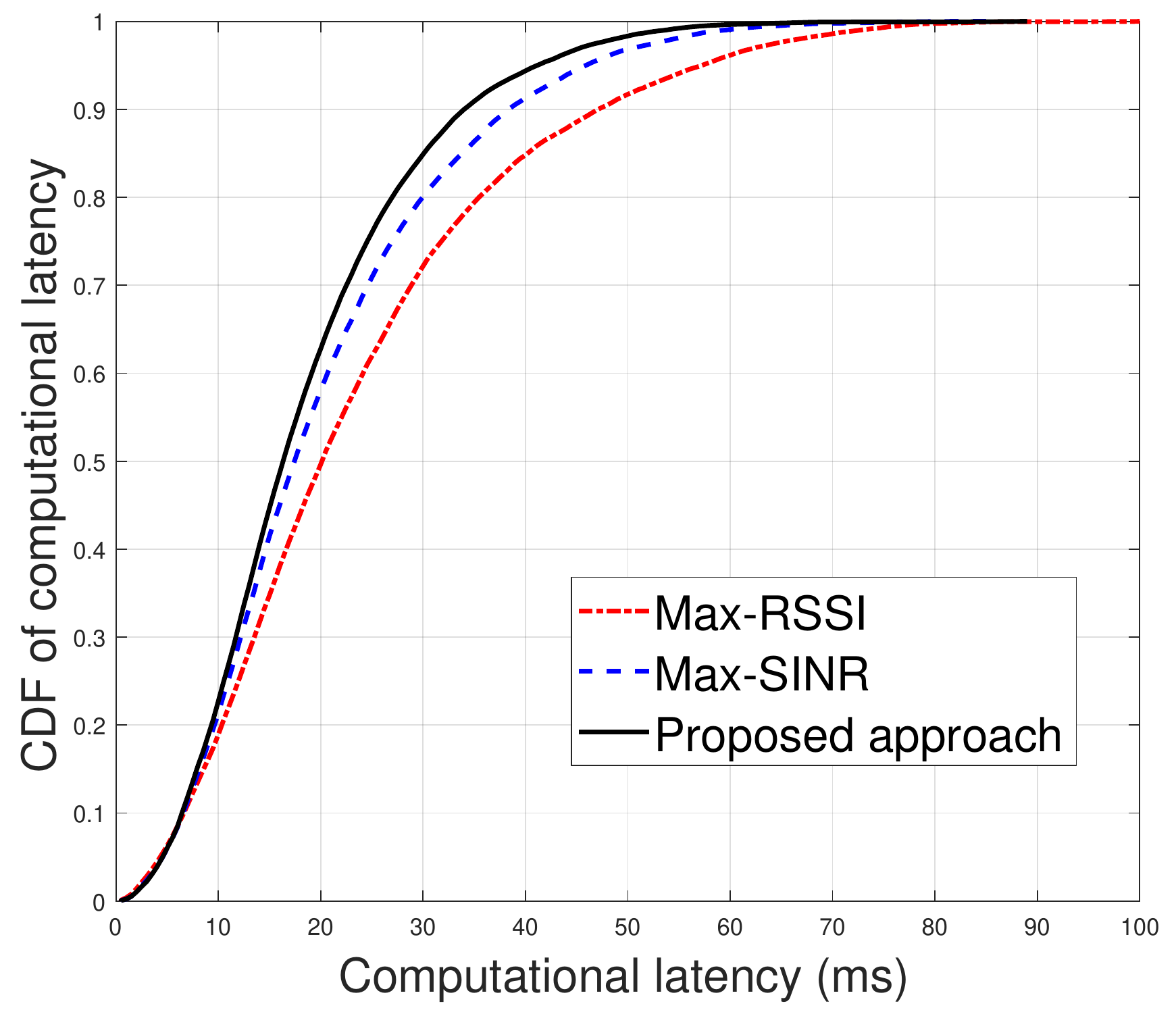}}\vspace{-0.1cm}
	\caption{CDF of the computational latency.}\vspace{-.4cm}
	\label{fig4}
\end{figure}

\begin{figure}[t!]
	\centering
	\centerline{\includegraphics[width=6.2cm]{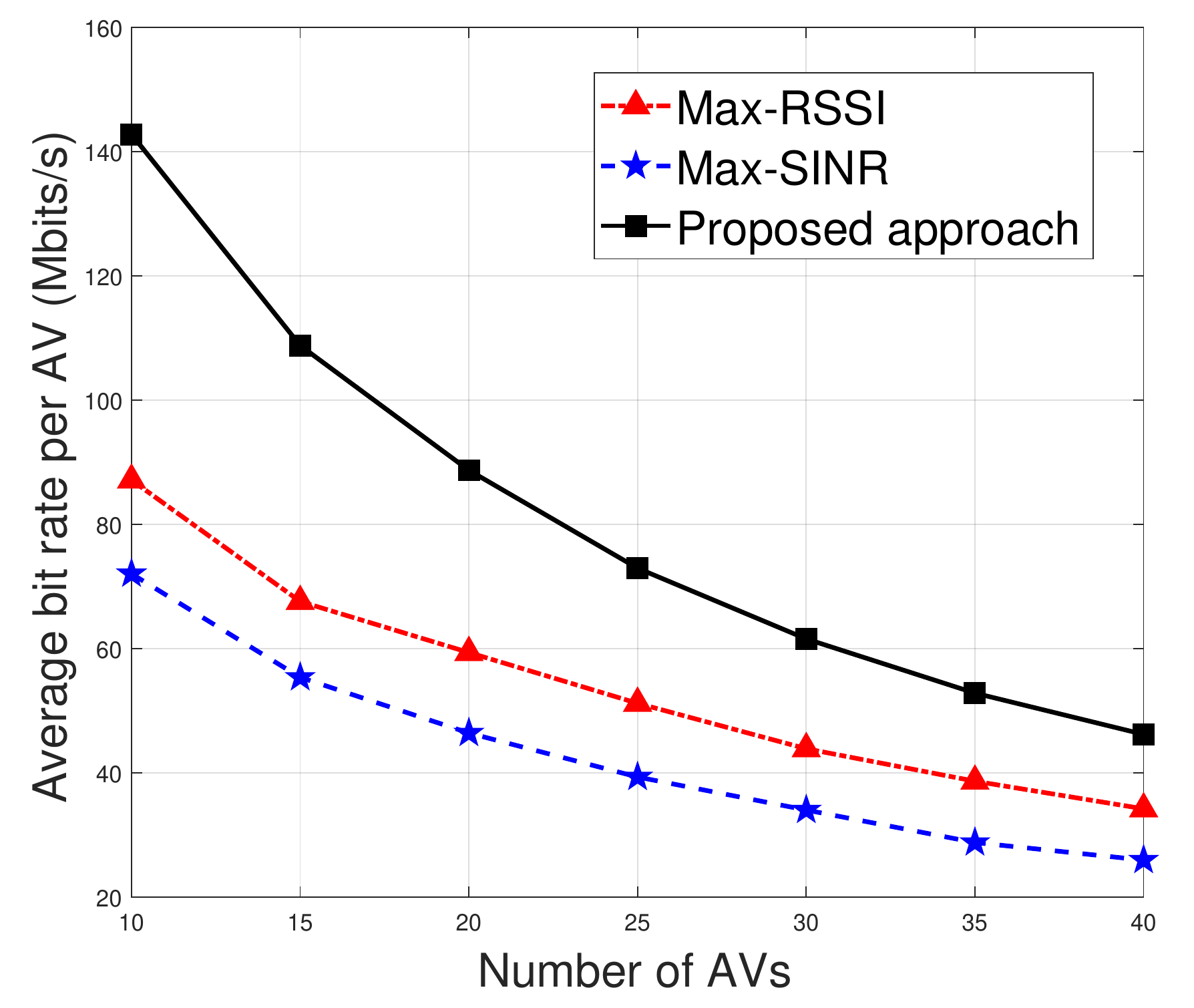}}\vspace{-0.1 cm}
	\caption{Average downlink rate versus the number of AVs.}\vspace{-.3cm}
	\label{fig5}
\end{figure}

\begin{figure}[t!]
	\centering
	\centerline{\includegraphics[width=6.2cm]{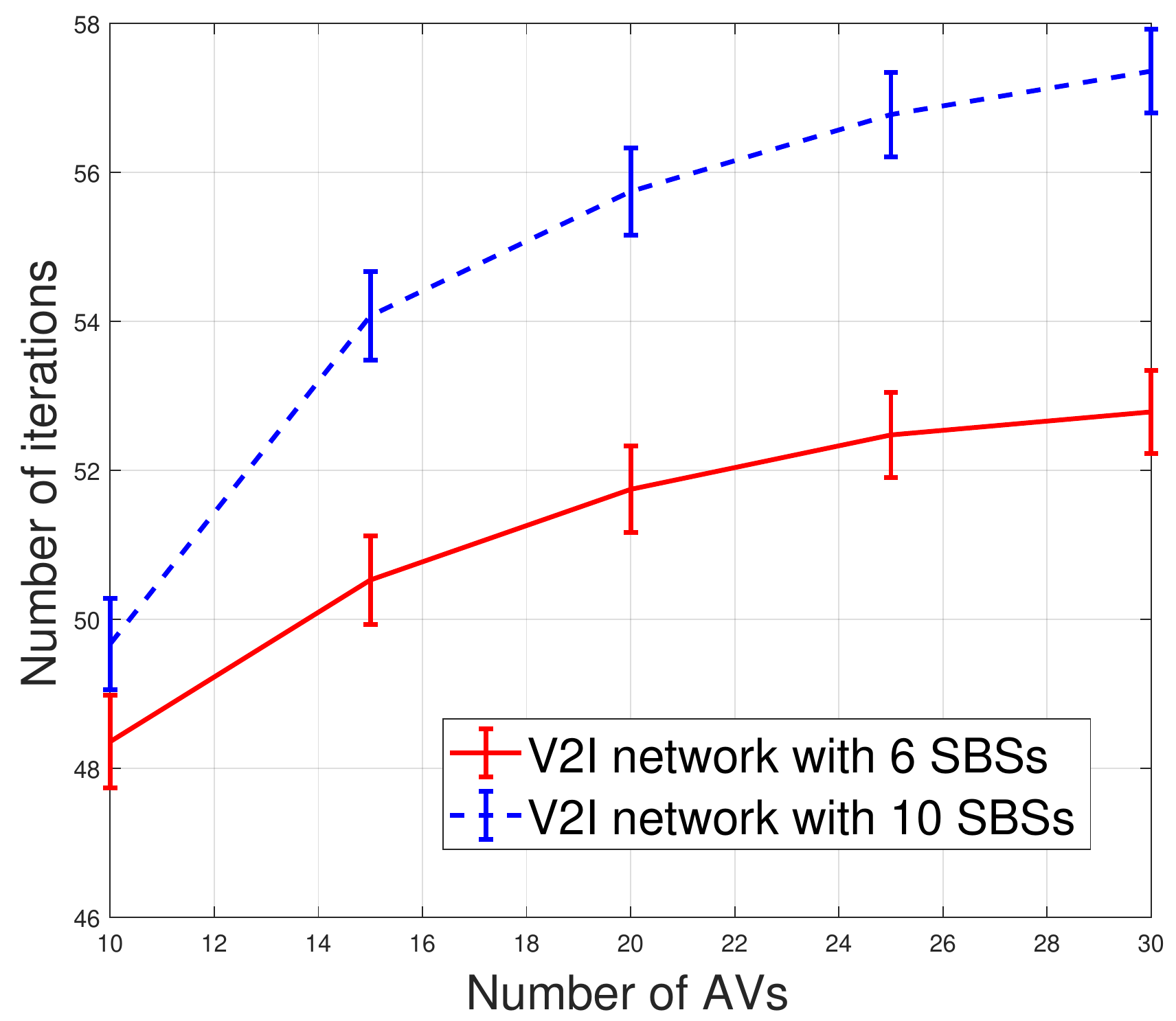}}\vspace{-0cm}
	\caption{Number of iterations versus the network size.}\vspace{-.5cm}
	\label{fig6}
\end{figure}

\section{Conclusions}
In this paper, we have proposed a novel framework for ultra reliable, low latency vehicles-to-infrastructure communications for autonomous vehicles. We have shown that the proposed framework can maximize the V2I network reliability, by jointly accounting for the interdependent computational delays for AVs, along with the transmission latency in the wireless network. In this regard, we have proposed a novel algorithm, based on the concept of labor matching markets, that allows distributed association of AVs with SBSs, while taking into account  the limited computational and bandwidth resources of each SBS. Furthermore, we have proved the convergence of the proposed algorithm to a core allocation of AVs to SBSs. Simulation results have shown the various merits of the proposed scheme.
\def\baselinestretch{.88}
\bibliographystyle{IEEEbib}
\bibliography{references,omid}
\end{document}